\documentclass{acm_proc_article-sp}
\usepackage[final,inline,nomargin]{fixme}   
\usepackage{graphicx}
\usepackage{amsmath}
\usepackage{stmaryrd}                       
\usepackage{dsfont}                         
\usepackage{bm}
\usepackage{amssymb}  

\newcommand{\tn}[1]{\textnormal{#1}}        
\newtheorem{theorem}{Theorem}
\newtheorem{definition}{Definition}
\newtheorem{proposition}{Proposition}

\newtheorem{example}{Example}
\newtheorem{remark}{Remark}
\newtheorem{corollary}{Corollary}
\newtheorem{procedure}{Procedure}
\newtheorem{lemma}{Lemma}

\begin{document}
%

\title{Timed Game Abstraction of Control Systems\titlenote{This work was supported by MT-LAB, a VKR Centre of Excellence.}}

\numberofauthors{2}
\author{
\alignauthor
Christoffer Sloth\\
       \affaddr{Department of Computer Science}\\
       \affaddr{Aalborg University}\\
       \affaddr{9220 Aalborg East, Denmark}\\
       \email{csloth@cs.aau.dk}
\alignauthor
Rafael Wisniewski\\
       \affaddr{Section of Automation \& Control}\\
       \affaddr{Aalborg University}\\
       \affaddr{9220 Aalborg East, Denmark}\\
       \email{raf@es.aau.dk}
}

\maketitle
\begin{abstract}
This paper proposes a method for abstracting control systems by timed game automata, and is aimed at obtaining automatic controller synthesis.

The proposed abstraction is based on partitioning the state space of a control system using positive and negative invariant sets, generated by Lyapunov functions. This partitioning ensures that the vector field of the control system is transversal to the facets of the cells, which induces some desirable properties of the abstraction. To allow a rich class of control systems to be abstracted, the update maps of the timed game automaton are extended.

Conditions on the partitioning of the state space and the control are set up to obtain sound abstractions. Finally, an example is provided to demonstrate the method applied to a control problem related to navigation.
\end{abstract}




%


\section{Introduction}\label{sec:introduction}
Controller design has been studied for many decades in the control community. In these studies, the primary objectives have been asymptotic stability and disturbance attenuation. This type of controller design is quite mature and can be used to synthesize controllers via, e.g., LMI-based method for linear systems. However, for nonlinear systems, design methods are limited and a considerable amount of manual labor is required in the controller design.

Controller design has also been considered in the computer science community for, e.g., discrete event systems and timed game automata. The requirements for such a system are primarily based on reachability of the system and temporal properties of the system. Especially, the timing requirements used in  computer science are very different from requirements known in control theory, as these are defined for a finite time horizon; whereas,  control theory is concerned with convergence, i.e., system properties when time goes to infinity. Fully automated tools have been developed for controller synthesis of discrete event systems and timed game automata. These are based on formal verification methods; therefore, the designed controllers are correct-by-design, i.e., the closed-loop control system is guaranteed to comply with the specification. This, in principle, eliminates the need for simulating the closed-loop control systems to perform further verification.

The goal of this paper is to abstract control systems by an automata-based model, and thereby allow automatic controller synthesis. In this way we are able to specify requirements in terms of Timed Computation Tree Logic (TCTL) specifications \cite{113766}; hence, requirements to reachability and timing can be added to the usual stability requirement.

Methods for synthesizing controllers for this kind of specification have been proposed in the computer science community for timed game automata \cite{Asarin98controllersynthesis}. For games the controller is called a strategy, and it decides among the possible choices in the game. The strategy can be automatically synthesized using tools such as UPPAAL Tiga \cite{Cassez05efficienton-the-fly}.

Methods from formal verification have already been adopted in control theory for controller design in \cite{Approximately_Bisimilar_Symbolic_Models_for_Incrementally_Stable_Switched_Systems}. Here the controller is synthesized to avoid certain unsafe states. For this purpose, a concept of approximate bisimulation, a relaxation of exact bisimulation, has been introduced.
This is further demonstrated in \cite{Fainekos07hierarchicalsynthesis}, where a robot is controlled to avoid some obstacles using a temporal logic specification. However, the generation of the models used for the synthesis procedure of these methods is based on simulating the system, which makes the method computationally demanding.

Methods for discretized models also exist, where solutions to the system equations are not utilized. One such method is presented in \cite{1582817}, where the principle of control to facet is utilized to synthesize a control strategy. 

In this paper, we abstract control systems by timed game automata with an extended update map, and use ideas similar to the bisimulation functions used in \cite{Approximately_Bisimilar_Symbolic_Models_for_Incrementally_Stable_Switched_Systems} in the abstraction procedure. However, the method does not require solutions to the system equations and is therefore not as computationally expensive. The work is an extension of the abstraction procedure presented in \cite{CDC2010}, which applies for autonomous systems.
This abstraction of dynamical systems by timed automata is based on partitioning the state space using Lyapunov functions. The intersections of sub-level sets of Lyapunov functions are used to form the cells that discretize the state space. This makes the problem of synthesizing a control strategy similar to control to facet, as the cells are generated using intersections of invariant sets. We provide a method for the design of switched controllers.
The main result of the paper is Theorem~\ref{prop:suf_cond_soundness}, which states sufficient and necessary conditions for sound abstractions by timed game automata.
Since we abstract the control systems by timed game automata with a modified update map, the synthesis procedure cannot yet be accomplished using existing tools.

This paper is organized as follows: Section~\ref{sec:preliminaries} contains preliminary definitions used throughout the paper, Section~\ref{sec:generation_of_finite_partition} explains the partitioning of the state space and the control, and Section~\ref{sec:obtaining_TA} describes how a timed game is generated from the partition. In Section~\ref{sec:partitioning_ss}, conditions for soundness are set up, Section~\ref{sec:results} provides an example, and Section~\ref{sec:conclusion} comprises conclusions.

\subsection{Notation}
The set $\{1,\dots,k\}$ is denoted $\bm{k}$. $B^{A}$
is the set of maps $A\rightarrow B$, $C(\mathds{R}^{n},\mathds{R}^{m})$ is the set of continuous maps $\mathds{R}^{n}\rightarrow \mathds{R}^{m}$. The power set of $A$ is denoted $2^{A}$. Given a vector $a\in\mathds{R}^{n}$, $a(j)$ denotes the $j^{\tn{th}}$ coordinate of $a$. Given a set $A$, the cardinality of the set is denoted $|A|$. Consider the Euclidean space $(\mathds{R}^{n},\langle,\rangle)$, where $\langle,\rangle$ is the scalar product. The state space is a connected subset $X\subset\mathds{R}^{n}$ such that there exists an open set $U$ such that $\tn{cl}(U)=X$. Whenever $f: X \to \mathds{R}$ is a function and $a \in \mathds{R}$, we write $f^{-1}(a)$ to shorten the notation of $f^{-1}(\{a\})$.

\section{Preliminaries}\label{sec:preliminaries}
The purpose of this section is to provide definitions related to control systems and timed game automata.

A control system $\Gamma=(X,U,f)$ has state space $X\subseteq\mathds{R}^{n}$, input space $U\subseteq\mathds{R}^{m}$, and dynamics described by ordinary differential equations $f:X\times U\rightarrow\mathds{R}^{n}$
\begin{align}
\dot{x}&=f(x,u).\label{eqn:general_control_system}
\end{align}
The input $u$, is controlled via a continuous map $g:X\rightarrow U$.

The system $\Gamma=(X,U,f)$ with the control $g$ is denoted $\Gamma_{g}=(X,f_{g})$, where
\begin{align}
\dot{x}&=f_{g}(x)=f(x,g(x)).\label{eqn:controlled_differential_equation}
\end{align}
We assume that $f_{g}$ is locally Lipschitz and has linear growth, then there exists a unique solution of \eqref{eqn:controlled_differential_equation} on $(-\infty,\infty)$ \cite{nonsmooth_analysis_and_control_theory}.

The solution of \eqref{eqn:controlled_differential_equation}, from an initial state $x_{0}\in X_{0}\subseteq X$ at time $t\geq0$ is described by the flow function $\phi_{\Gamma_{g}}:[0,\epsilon]\times X\rightarrow X$, $\epsilon>0$ satisfying
\begin{align}
\frac{d\phi_{\Gamma_{g}}(t,x_{0})}{dt}&=f_{g}\left(\phi_{\Gamma_{g}}(t,x_{0})\right)\label{eqn:solution_of_auto_differential_equation}
\end{align}
for all $t\geq0$.

Lyapunov functions are utilized in stability theory and are defined in the following \cite{Energy_Functions_for_Morse_Smale_Systems}.
\begin{definition}[Lyapunov Function]\label{def:lyapunov_function}
Let $X$ be an \\open connected subset of $\mathds{R}^{n}$. Suppose $f_{g}:X\rightarrow\mathds{R}^{n}$ is continuous and let $\tn{Cr}(f_{g})$ be the set of critical points of $f_{g}$.
Then a real non-degenerate differentiable function $\varphi:X\rightarrow\mathds{R}$ is said to be a Lyapunov function for $f_{g}$ if
\begin{subequations}
\begin{align}
\notag&p\tn{ is a critical point of }f_{g}\Leftrightarrow p\tn{ is a critical point of }\varphi\\
&\dot{\varphi}_{g}(x)\equiv\sum_{j=1}^{n}\frac{\partial \varphi}{\partial x_{j}}(x)f_{g}^{j}(x)\label{eqn:Lyap_der}\\
&\dot{\varphi}_{g}(x)=0\,\,\,\,\,\,\,\,\forall x\in\tn{Cr}(f_{g})\\
&\dot{\varphi}_{g}(x)\neq0\,\,\,\,\,\,\,\,\forall x\in X\backslash\tn{Cr}(f_{g})
\end{align}
\end{subequations}
and there exists $\alpha>0$ and an open neighborhood of each critical point $p\in\tn{Cr}(f_{g})$, where
\begin{align}
&||\dot{\varphi}_{g}(x)||\geq\alpha||x-p||^{2}.
\end{align}
\end{definition}
\begin{remark}
We only require the vector field to be transversal to the level curves of a Lyapunov function $\varphi$, i.e., $\dot{\varphi}_{g}(x)=\langle\nabla \varphi_{g} (x), f_{g}(x)\rangle\neq0$ for all $x\in X\backslash\tn{Cr}(f_{g})$, and does not use Lyapunov functions in the usual sense, where the existence of a Lyapunov function implies stability, but uses a more general notion from \cite{Energy_Functions_for_Morse_Smale_Systems}.
\end{remark}
To simplify the notation, we use subscript $g$ on $\dot{\varphi}_{g}$ to indicate that the control $g$ is applied in the calculation of $\dot{\varphi}_{g}$.

\begin{definition}[Reachable set of Dyn. System]$ $\\
The reachable set of a dynamical system $\Gamma_{g}$ from a set of initial states $X_{0}\subseteq X$ on the time interval $[t_{1},t_{2}]$ is defined as
\begin{align}
\notag\tn{Reach}_{[t_{1},t_{2}]}(\Gamma_{g},X_{0})\equiv\{&x\in X|\exists t\in[t_{1},t_{2}]\tn{, }\exists x_{0}\in X_{0}\\
&\tn{such that }x =\phi_{\Gamma_{g}}(t,x_{0})\}.
\end{align}
\end{definition}

The control system will be abstracted by a timed game automaton, which is an extension of a timed automaton \cite{Verification_Performance_Analysis_and_Controller_Synthesis_for_RealTime_Systems}. In the definition of a timed automaton, a set of diagonal-free clock constraints $\Psi(C)$ is used for the set $C$ of clocks. $\Psi(C)$ is defined as the set of constraints $\psi$ described by the following grammar: 
\begin{align}
&\psi::=c\bm{\bowtie} k|\psi_{1}\bm{\wedge}\psi_{2}\tn{, where}\label{eqn:clock_grammar}\\
&\notag c\in C,\, k\in\mathds{R}_{\geq0},\tn{ and }\bm{\bowtie}\in\{\bm{\leq},\bm{<},\bm{=},\bm{>},\bm{\geq}\}.
\end{align}
Note that the clock constraint $k$ should usually be an integer, but in this paper no effort is done to convert the clock constraints into integers. Furthermore, the elements of $\bm{\bowtie}$ are bold to indicate that they are syntactic operations.
\begin{definition}[Timed Automaton]
A timed automaton, $\mathcal{A}$, is a tuple $(E, E_{0}, C, \Sigma, I, \Delta)$, where
\begin{itemize}
\item $E$ is a finite set of locations, and $E_{0}\subseteq E$ is the set of initial locations.
\item $C$ is a finite set of clocks.
\item $\Sigma$ is the set of actions.
\item $I:E\rightarrow\Psi(C)$ assigns invariants to locations.
\item $\Delta\subseteq E\times\Psi(C)\times\Sigma\times2^{C}\times E$ is a finite set of transition relations. The transition relations provide edges between locations as tuples $(e,G_{e\rightarrow e'},\sigma,R_{e\rightarrow e'},e')$, where $e$ is the source location, $e'$ is the destination location, $G_{e\rightarrow e'}\in\Psi(C)$ is the guard set, $\sigma$ is an action in $\Sigma$, and $R_{e\rightarrow e'}\in2^{C}$ is the set of clocks to be reset.
\end{itemize}
\end{definition}
The semantics of a timed automaton is defined in the following, adopting the notion of \cite{Verification_Performance_Analysis_and_Controller_Synthesis_for_RealTime_Systems}.
\begin{definition}[Clock Valuation]
A clock valuation on a set of clocks $C$ is a mapping $v:C\rightarrow\mathds{R}_{\geq0}$. The initial valuation $v_{0}$ is given by $v_{0}(c)=0$ for all $c\in C$. For a valuation $v$, $t\in\mathds{R}_{\geq0}$, and $R\subseteq C$, the valuations $v+t$ and $v[R]$ are defined as follows
\begin{subequations}
\begin{align}
(v+t)(c) &= v(c)+t,\label{eqn:clock_valuation_delay}\\
v[R](c) &= \begin{cases}0&\tn{for }c\in R,\\v(c)&\tn{otherwise.}\end{cases}\label{eqn:clock_valuation_reset}
\end{align}
\end{subequations}
\end{definition}
We see that \eqref{eqn:clock_valuation_delay} is used to progress time and that \eqref{eqn:clock_valuation_reset} is used to reset the clocks in $R$ to zero.
\begin{definition}[Semantics of Clock Constraint]$ $\\
A clock constraint in $\Psi(C)$ is a set of clock valuations $\{v:C\rightarrow\mathds{R}_{\geq0}\}$ given by
\begin{subequations}
\begin{align}
\llbracket c\bm{\bowtie} k\rrbracket &= \{v:C\rightarrow\mathds{R}_{\geq0}|v(c)\bowtie k\}\\
\llbracket\psi_{1}\bm{\wedge}\psi_{2}\rrbracket &= \llbracket\psi_{1}\rrbracket\cap\llbracket\psi_{2}\rrbracket.
\end{align}
\end{subequations}
\end{definition}
For convenience we denote $v\in\llbracket\psi\rrbracket$ by $v\models\psi$.

\begin{definition}[Semantics of Timed Automaton]$ $\\
The semantics of a timed automaton $\mathcal{A}=(E, E_{0}, C, \Sigma, I, \Delta)$ is a transition system $\llbracket \mathcal{A}\rrbracket=(S,S_{0},\Sigma\cup\mathds{R}_{\geq0},T_{\sigma}\cup T_{\tn{t}})$, where
\begin{align}
\notag S =\, &\{(e,v)\in E\times\mathds{R}^{C}_{\geq0}|v\models I(e)\}\\
\notag S_{0} =\, &\{(e,v)\in E_{0}\times v_{0}\}\\
\notag T_{\sigma} =\, &\{(e,v)\overset{\sigma}{\rightarrow}(e',v')|\exists (e,G_{e\rightarrow e'},\sigma,R_{e\rightarrow e'},e')\in\Delta :\\\notag&v\models G_{e\rightarrow e'} \tn{, }v'=v[R_{e\rightarrow e'}]\}\\
\notag T_{\tn{t}} =\, &\{(e,v)\overset{t}{\rightarrow}(e,v+t)|\forall t'\in[0,t] : v+t'\models I(e)\}.
\end{align}
\end{definition}

Analog to the solution of \eqref{eqn:controlled_differential_equation}, shown in \eqref{eqn:solution_of_auto_differential_equation}, is a run of a timed automaton.
\begin{definition}[Run of Timed Automaton]\label{def:run_of_timed_automaton}
A run \\of a timed automaton $\mathcal{A}$ is a possibly infinite sequence of alternations between time steps and discrete steps on the following form
\begin{align}
\varrho_{\mathcal{A}}:
(e_{0},v_{0})\overset{t_{1}}{\longrightarrow}(e_{0},v_{1})\overset{\sigma_{1}}{\longrightarrow}(e_{1},v_{2})\overset{t_{2}}{\longrightarrow}\dots
\end{align}
which is a path in $\llbracket\mathcal{A}\rrbracket$, where $t_{i}\in \mathds{R}_{\geq0}$ and $\sigma_{i}\in \Sigma$.
The multifunction describing the runs of a timed automaton $\phi_{\mathcal{A}}:\mathds{R}_{\geq0}\times E_{0}\rightarrow 2^{E}$, is defined by $e\in\phi_{\mathcal{A}}(t,e_{0})$ if and only if there exists a path in $\llbracket\mathcal{A}\rrbracket$ initialized in $(e_{0},v_{0})$ that reaches the location $e$ at time $t=\sum_{i} t_{i}$.
\end{definition}

From the run of a timed automaton, the reachable set is defined below.
\begin{definition}[Reachable set of Timed Auto.]$ $\\
The reachable set of a timed automaton $\mathcal{A}$, with initial locations $E_{0}$, in the time interval $[t_{1},t_{2}]$ is defined as
\begin{align}
\notag\tn{Reach}_{[t_{1},t_{2}]}(\mathcal{A},E_{0})\equiv\{&e\in E|\exists t\in[t_{1},t_{2}],\exists e_{0}\in E_{0}\\
&\tn{such that }e \in\phi_{\mathcal{A}}(t,e_{0})\}.
\end{align}
\end{definition}

A timed game automaton is closely related to a timed automaton as shown in the following \cite{A_Game-Theoretic_Approach_to_Real-Time_System_Testing}.
\begin{definition}[Timed Game Automaton]\label{def:timed_game_automaton}
A timed \\game automaton  is a tuple $\mathcal{G}=(E, E_{0}, C, \Sigma_{\tn{c}},\Sigma_{\tn{u}}, I, \Delta)$, \\where the tuple $(E, E_{0}, C, \Sigma_{\tn{c}}\cup\Sigma_{\tn{u}}, I, \Delta)$ is a timed automaton. The set $\Sigma_{\tn{c}}$ contains controllable actions (system inputs) and the set $\Sigma_{\tn{u}}$ contains uncontrollable actions (system outputs).
\end{definition}
Actions that can be affected by a strategy, see Definition~\ref{def:strategy}, are controllable actions. However, transitions labeled with uncontrollable actions can happen whenever an adversary (the environment) chooses to take them and the associated guards are satisfied. In this setting, the input actions are equivalent to changing the control $g(x)$, and the output actions are observations, i.e., they resemble information about the state of the system.


\begin{example}
Consider the timed game automaton shown in Figure~\ref{fig:TGA} having four locations and one clock. It is initialized in location $e_{1}$ and the aim is to reach location $e_{3}$. Initially, an uncontrollable action can enable a transition to $e_{2}$, when the guard $c\bm{\geq}1$ is satisfied. After reaching $e_{2}$, the game is guaranteed to reach location $e_{3}$ if the controllable action $\sigma_{\tn{c}}$ happens when $c\bm{<}1$. Otherwise a transition to $e_{4}$ may occur.
\begin{figure}[!htb]
    \centering
       \includegraphics[scale=.9]{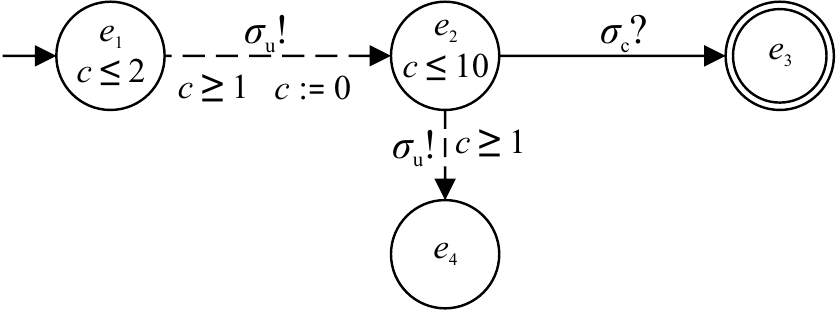}
    \caption{Timed game automaton with four locations, where the solid line represents a transition with a controllable action and the dashed lines represent transitions with uncontrollable actions.}
    \label{fig:TGA}
\end{figure}
\end{example}
Analog to the controller $g(x)$, a strategy is defined in the following. Let us first define continuation of a run.
\begin{definition}[Continuation of a Run]
Let $\varrho=$ \\$(e_{0},v_{0})\rightarrow\dots\rightarrow(e_{k},v_{k})$ be a run.
Then a continuation of the run $\varrho$ is the run $\varrho'=(e_{0},v_{0})\rightarrow\dots\rightarrow(e_{k},v_{k})\rightarrow(e_{k+1},v_{k+1})$, denoted $\varrho'=\varrho \to (e_{k+1}, v_{k+1})$.
\end{definition}

\begin{definition}[Strategy]\label{def:strategy}
Let $\mathcal{G}=(E, E_{0}, C, \Sigma_{\tn{c}},$\\$\Sigma_{\tn{u}}, I, \Delta)$ be a timed game automaton. Then any map
$\kappa:S^{+}\rightarrow\Sigma_{c}\cup\{\delta\}$, where $S^{+}$ is the set of all runs and $\delta\notin\Sigma_{\tn{c}}\cup\Sigma_{\tn{u}}$, which  satisfies the following two conditions:
\begin{enumerate}
\item if $\kappa(\varrho)=\delta$, then $\varrho\overset{t}{\rightarrow}(e_{k},v_{k}+t)$ is a path in $\llbracket\mathcal{G}\rrbracket$ for some $t>0$, and
\item if $\kappa(\varrho)=\sigma$, then $\varrho\overset{\sigma}{\rightarrow}(e_{k}',v_{k}')$ is a path in $\llbracket\mathcal{G}\rrbracket$,
\end{enumerate}
for any run $\varrho=(e_{0},v_{0})\rightarrow\dots\rightarrow(e_{k},v_{k})$, is called a strategy.
\end{definition}
We see that the controller can either do nothing ($\kappa(\varrho)=\delta$) or execute a controllable action ($\kappa(\varrho)=\sigma$). Furthermore, we see that a strategy may depend on the entire past of the run. However, we are only interested in so-called memoryless strategies, i.e., strategies that only depend on the current state of the timed game automaton. This implies that if $\varrho=(e_{0},v_{0})\overset{t_{1}}{\rightarrow}(e_{1},v_{1})\rightarrow\dots\rightarrow(e_{k},v_{k})$ and $\varrho'=(e_{0},v_{0})\overset{t_{1}'}{\rightarrow}(e_{1}',v_{1}')\rightarrow\dots\rightarrow(e_{k},v_{k})$, then $\kappa(\varrho)=\kappa(\varrho')$. In addition to being memoryless, we require the strategy to be independent of the clock valuations, i.e., if $\varrho=(e_{0},v_{0})\overset{t_{1}}{\rightarrow}\dots\rightarrow(e_{k},v_{k})$ and $\varrho'=(e_{0},v_{0})\overset{t_{1}'}{\rightarrow}(e_{1}',v_{1}')\rightarrow\dots\rightarrow(e_{k},v_{k}')$, then $\kappa(\varrho)=\kappa(\varrho')$. A closed-loop system satisfying these properties can be implemented as a parallel composition of a timed automaton and an automaton.

\begin{definition}\label{def:closed_loop_game}
Consider the timed game automaton $\mathcal{G}$ and the strategy $\kappa$, which only depends on the locations. Then $\mathcal{G}_{\kappa}$ is the timed automaton controlled using the strategy $\kappa$.
\end{definition}

\section{Generation of Finite Partition}\label{sec:generation_of_finite_partition}
The presented abstraction is based on partitioning the state space and the control of $\Gamma=(X,U,f)$. The partitioning of the state space is inspired by the partitioning presented in \cite{CDC2010}. However, in contrast to \cite{CDC2010}, the considered system has an unknown control of $u$, as stated in \eqref{eqn:general_control_system}.

It is proposed to partition the state space by intersecting slices defined as the set-difference of positive and negative invariant sets. This implies that the partition should be conducted such that for each admissible control $g^{i}(x)$ the vector field of the controlled system $\Gamma_{g^{i}}=(X,f_{g^{i}})$ is transversal to the boundaries of the slices. The admissible controls are defined to be the finite set $K_{U}=\{g^{i}(x)|i\in\Lambda\}$, where $\Lambda$ is some index set. For simplicity, the controls in $K_{U}$ have domain $X$. However, this simplification can easily be relaxed.

\begin{definition}[Slice]\label{def:slice}
A nonempty set $S$ is a slice if it is a union of cells and there exist two sets $A$ and $B$ that are positive or negative invariant such that
\begin{enumerate}
\item $B$ is a proper subset of $A$, i.e., $B\subset A$.
\item Given any $g\in K_{U}$, $A$ and $B$ are either positive or negative invariant sets for  system $\Gamma_{g}=(X,f_{g})$, and
\item $S=\tn{cl}(A\backslash B)$.
\end{enumerate}
\end{definition}
From this definition, we see that for a given partition, only controls that make the vector field of the closed-loop system transversal to the boundaries of the slices are allowed.

To devise a partition of a state space, we need to define collections of slices, called slice-families.
\begin{definition}[Slice-Family]
A slice-family $\mathcal{S}$ is a collection of slices generated by the positive and negative invariant open sets $A_{1}\subset A_{2}\subset\dots\subset A_{k}$ covering the entire state space of $\Gamma$, thereby $S_{1}=A_{1}$, $S_{2}=\tn{cl}(A_{2}\backslash A_{1}),\dots,S_{k}=\tn{cl}(A_{k}\backslash A_{k-1})$ and $X\subseteq A_{k}$. For convenience $|\mathcal{S}|$ is defined to be the cardinality of the slice-family $\mathcal{S}$, thus $\mathcal{S}=\{S_{1},\dots,$\\$S_{|\mathcal{S}|}\}$. Furthermore, we say that $\mathcal{S}$ is generated by the sets $\{A_{i}|i\in\bm{k}\}$.
\end{definition}

A function is associated to each slice-family $\mathcal{S}$, to provide an easy way of describing the boundary of a slice. Such functions are called partitioning functions.
\begin{definition}[Partitioning Function]\label{def:part_fun}
Let $\mathcal{S}$ be a slice-family, then a continuous function $\varphi:\mathds{R}^{n}\rightarrow\mathds{R}$ smooth on $\mathds{R}^{n}\backslash\{0\}$ is a partitioning function associated to $\mathcal{S}$ if for any set $A_{i}$ generating $\mathcal{S}$ there exists $a_{i},a_{i}'\in\mathds{R}\cup\{-\infty,\infty\}$ such that
\begin{align}
\varphi^{-1}([a_{i}',a_{i}])=A_{i}
\end{align}
and $a_{i},a_{i}'$ are regular values of $\varphi$. By regular level set theorem, the boundary $\varphi^{-1}(a_{i})$ of $A_{i}$ is a smooth manifold \cite{An_Introduction_to_Manifolds}.
\end{definition}
In the remainder of the paper, we associate to each slice-family $\mathcal{S}^{i}$ a partitioning function $\varphi^{i}$.

The partition of the state space is associated with cells that are generated by intersecting slices.
\begin{definition}
We say that slices $S_{1}$ and $S_{2}$ intersect each other transversally and write
\begin{align}
S_{1}\pitchfork S_{2}=S_{1}\cap S_{2},
\end{align}
if their boundaries, $\tn{bd}(S_{1})$ and $\tn{bd}(S_{2})$, intersect each other transversally.
\end{definition}

\begin{definition}[Extended Cell]\label{def:extended_cell}
Let $\{\mathcal{S}^{i}|i\in\bm{k}\}$ be a collection of $k$ slice-families and define $\mathcal{Y}=\{1,\dots,|\mathcal{S}^{1}|\}\times\dots\times\{1,\dots, |\mathcal{S}^{k}|\}$. Denote the $j^{\tn{th}}$ slice in $\mathcal{S}^{i}$ by $S^{i}_{j}$ and let $y\in\mathcal{Y}$. Then
\begin{align}
e_{\tn{ex},y}&\equiv\pitchfork_{i=1}^{k} S^{i}_{y_{i}}.\label{eqn:extended_cell_def}
\end{align}
Any nonempty set $e_{\tn{ex},y}$ will be called an extended cell. These cells are denoted extended cells, since the transversal intersection of slices may form multiple disjoint sets.
\end{definition}

\begin{definition}[Cell]\label{def:cell}
A cell is a connected component of an extended cell
\begin{subequations}
\begin{align}
\bigcup_{z} e_{(y,z)} &= e_{\tn{ex},y}\tn{, where }\\
e_{(y,z)}\cap e_{(y,z')}&=\emptyset\,\,\,\,\,\,\,\,\forall z\neq z'.
\end{align}
\end{subequations}
\end{definition}
We say that the slices $S^{1}_{y_{1}},\dots,S^{k}_{y_{k}}$ generate the cell. In the remainder of the paper, we denote the slice from the $i^{\tn{th}}$ slice-family generating $e$ by $S_{e}^{i}$.

\begin{proposition}[Proof in \cite{CDC2010_proofs}]
If $S_{1}\pitchfork S_{2}\neq\emptyset$ then
\begin{align}
&\tn{int}(S_{1}\cap S_{2})\neq\emptyset.
\end{align}
\end{proposition}

A finite partition of the state space based on the transversal intersection of slices is defined in the following.
\begin{definition}[Finite Partition of State Space]\label{def:finite_partition}
Let $\mathcal{S}$ be a collection of slice-families, $\mathcal{S}=\{\mathcal{S}^{i}|i\in\bm{k}\}$. Then the finite partition $K_{X}(\mathcal{S})$ is defined to be the collection of all cells generated by $\mathcal{S}$ according to Definition~\ref{def:cell}.
\end{definition}
Finally, the product of $K_{X}(\mathcal{S})$ and $K_{U}$ is defined to be the partition of $\Gamma$ given by $K(\mathcal{S})\equiv K_{X}(\mathcal{S})\times K_{U}$.

The following example clarifies how the partitioning is conducted.
\begin{example}\label{example:TGA_partition_example}
Consider the one-dimensional system with one control input and state space $X=[-3,3]\subset\mathds{R}$
\begin{align}
\dot{x}=-x+u.\label{eqn:partition_example}
\end{align}
The partition of its state space is conducted using the sets $A_{1}=[-1,1]$ and $A_{2}=[-3,3]$, i.e., we obtain the cells: $[-3,-1]$, $[-1,1]$, $[1,3]$. The control of the system should also be partitioned, which is chosen to $K_{U}=\{0,1.5,2x\}$. Figure~\ref{fig:partition_example} shows a partition associated with \eqref{eqn:partition_example}.
\begin{figure}[!htb]
    \centering
       \includegraphics[scale=1]{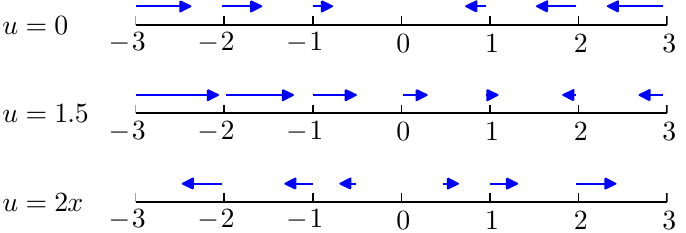}
    \caption{Vector field of \eqref{eqn:partition_example}, illustrated with blue arrows, with three different controls applied to the system.}
    \label{fig:partition_example}
\end{figure}

From the figure it is seen that the direction of the vector field can be reversed completely or partly according to the input applied to the system. Furthermore, the two sets $A_{1}=[-1,1]$ and $A_{2}=[-3,3]$ are positive or negative invariant sets for all controls in $K_{U}$.
\end{example}

\section{Generation of Timed Game from Finite Partition}\label{sec:obtaining_TA}
To abstract a control system $\Gamma$ by a timed game automaton $\mathcal{G}$, we modify the abstraction procedure presented in \cite{CDC2010}, by adding the distinction between controllable and uncontrollable actions according to Definition~\ref{def:timed_game_automaton}. Furthermore, we extend the expressiveness of the update maps of the timed game, to allow more accurate abstractions. First, an example is provided to illustrate the principle of the abstraction, then the abstraction procedure is presented.

The possibility to change the control input increases the number of locations and adds the possibility to influence the trajectories of the system, as shown in the following example.
\begin{example}\label{example:TGA_ctrl_reset_example}
$ $ $ $ Consider a one-dimensional system from \\Example~\ref{example:TGA_partition_example}, where $X=[-3,3]$ and
\begin{align}
\dot{x}=-x+u.
\end{align}
The state space $X$ is partitioned into three cells, i.e., $K_{X}=\{[-3,-1],[-1,1],[1,3]\}$ and the controls are in $K_{U}=\{0,1.5,$\\$2x\}$. All controls can be applied in all cells, i.e., the generated game has 9 locations ($K_{X}\times K_{U}$).
A timed game abstracting this system is illustrated in Figure~\ref{fig:TGA_ctrl_reset_example}, where the execution of the controllable action $\sigma_{\tn{c}}^{i}$ resembles applying the control $u=g^{i}(x)$ in the dynamical system.
\begin{figure}[!htb]
    \centering
       \includegraphics[width=\linewidth]{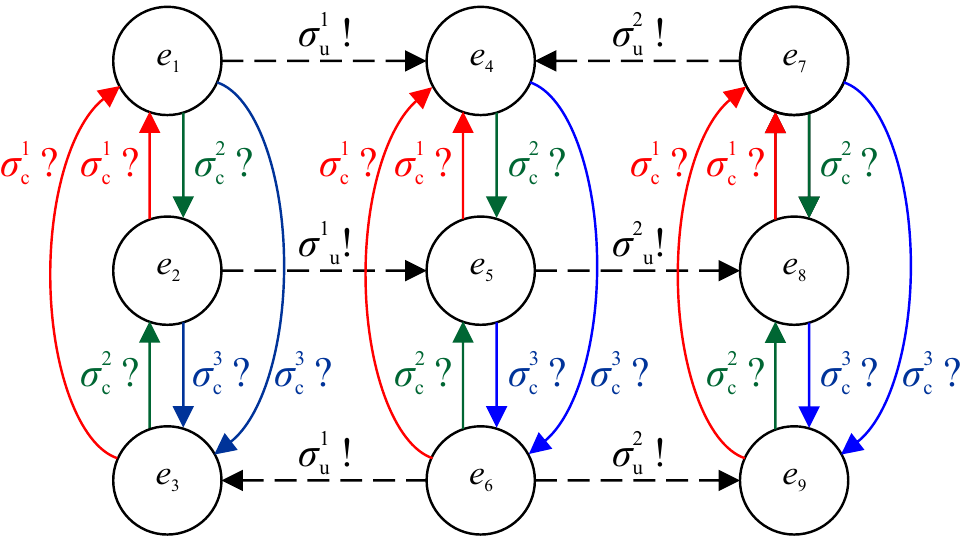}
    \caption{Illustration of a timed game automaton, where three different controls can be applied to the system.}
    \label{fig:TGA_ctrl_reset_example}
\end{figure}
\end{example}

Before providing the procedure for generating a timed game automaton, the clock valuation is redefined, as a more expressive update map is used in order to abstract the systems. Furthermore, we define the clock valuation on pairs of clocks as the abstraction uses pairs of clocks to monitor the fastest respectively slowest progress in each direction.
\begin{definition}[Extended Clock Valuation]$ $\\
Denote the pair of clocks $(c_{1}^{i},c_{2}^{i})$ by $c^{i}$, and let $C=\{c^{i}|i\in\bm{k}\}$. Then a clock valuation on a set of clocks $C$ is a mapping $v:C\rightarrow\mathds{R}^{2}_{\geq0}$. The initial valuation $v_{0}$ is given by $v_{0}(c)=(0,0)$ for all $c\in C$. For a valuation $v$, $t\in\mathds{R}_{\geq0}$, and $R\subseteq C$, the valuations $v+t$ and $v[R]$ on a pair of clocks $c$ are
\begin{subequations}
\begin{align}
(v+t)(c) &= v(c)+\begin{pmatrix}1\\1\end{pmatrix}t,\label{eqn:clock_valuation_delay_extended}\\
v[R](c) &= \alpha+\beta v(c)\tn{ for }c\in R\label{eqn:clock_valuation_update_extended}
\end{align}
\end{subequations}
where $\alpha\in\mathds{R}^{2}$ and $\beta$ is a $2\times2$ matrix with rational entries.
\end{definition}
Note that resetting a pair of clocks is just a special case of \eqref{eqn:clock_valuation_update_extended}, where $\alpha=(0,0)$ and all entries of $\beta$ are zero.
Furthermore, the valuation of one clock, denote $v(c_{1}^{i})$ or $v(c_{2}^{i})$, is also used in the following, when appropriate.

In the following, a procedure is presented for generating a timed game automaton from a finite partition. 
\begin{procedure}\label{procedure:generation_of_TGA}
Given a partition $K(\mathcal{S})$, the timed game automaton $\mathcal{G}=(E, E_{0}, C, \Sigma_{\tn{c}},\Sigma_{\tn{u}}, I, \Delta)$ is generated using the following
\begin{itemize}
\item \textbf{Locations:} Let the locations of $\mathcal{G}$ be given by
\begin{align}
E=K_{X}(\mathcal{S})\times K_{U}.\label{eqn:TA_locations}
\end{align}
We use the notation $e_{(y,z,j)} \equiv (e_{(y,z)},g^j) \in E$.
This means that a location $e_{(y,z,j)}$ is associated with the cell $e_{(y,z)}$ of the partition $K_{X}(\mathcal{S})$ and the control $g^{j}(x)$ in $K_{U}$.
\item \textbf{Clocks:} Given $k$ slice-families, the number of clocks is $2k$, i.e., $C=\{c^{i}|i\in\bm{k}\}$. The pair of clocks $c^{i}=(c^{i}_{1},c^{i}_{2})$ monitors the maximum and minimum time for being in slices of the slice-family $\mathcal{S}^{i}$.
\item \textbf{Invariants:} In each location $e\in E$, there are up to $k$ invariants. The invariants for location $e_{(y,z,j)}$ specify upper bounds on the time for staying in the $k$ slices generating the cell $e_{(y,z)}$ with a control $g^{j}$ applied in $e_{(y,z)}$. We say that a cell $e_{(y,z)}$ is generated by the slices $\{S^{i}_{y_{i}}|i\in\bm{k}\}$ and in addition we say that a location $e_{(y,z,j)}$ is generated by $\{S^{i}_{(y_{i},g_{j})}|i\in\bm{k}\}$ and use the shorthand notation $S^{i}_{e}\equiv S^{i}_{(y_{i},g_{j})}$ for convenience. We impose an invariant whenever there is an upper bound on the time for staying in a slice generating the cell $e$
\begin{align}
I(e)&= \bigwedge_{i=1}^{k}c^{i}_{1}\bm{\leq} \overline{t}_{S^{i}_{e}}\label{eqn:TGA_invariants}
\end{align}
where $\overline{t}_{S^{i}_{e}}\in\mathds{R}_{\geq0}$ is an upper bound on the time for staying in $S^{i}_{e}$.
\item \textbf{Controllable Actions:} The controllable actions $\Sigma_{\tn{c}}$ are actions $\sigma_{\tn{c}}^{1},\dots,\sigma_{\tn{c}}^{|K_{U}|}$, where $\sigma_{\tn{c}}^{j}$ is associated with applying the control $g^{j}(x)$ to the dynamical system.
\item \textbf{Uncontrollable Actions:}
The uncontrollable actions $\Sigma_{\tn{u}}$ are actions $\sigma_{\tn{u}}^{1},\dots,\sigma_{\tn{u}}^{k}$, where $\sigma_{\tn{u}}^{i}$ is associated with transitions between slices of the $i^{\tn{th}}$ slice-family $\mathcal{S}^{i}=\{S^{i}_{1},\dots,S^{i}_{|\mathcal{S}^{i}|}\}$.
\item \textbf{Transition relations:} If a pair of locations, $e$ and $e'$, where the control $g(x)$ is applied in both $e$ and $e'$, satisfy the following two conditions
\begin{enumerate}
\item $e$ and $e'$ are adjacent cells in the state space, i.e., $e\cap e'\neq\emptyset$, with $S^{i}_{e}\neq S^{i}_{e'}$ for some $i\in\bm{k}$. Hence, $e$ and $e'$ are generated from different slices in $\mathcal{S}^{i}$, and
\item $\varphi^{i}(x')\leq\varphi^{i}(x)\,\forall x\in e$ and $\forall x'\in e'$ and $\dot{\varphi}^{i}_{g}(x)<0$ $\forall x\in e\cap e'$ or $\varphi^{i}(x')\geq\varphi^{i}(x)\,\forall x\in e$ and $\forall x'\in e'$ and $\dot{\varphi}^{i}_{g}(x)>0$ $\forall x\in e\cap e'$.
\end{enumerate}
Then there is a transition relation
\begin{subequations}
\begin{align}
&\delta_{e\rightarrow e'} = (e, G_{e\rightarrow e'}, \sigma_{\tn{u}}^{i}, R_{\tn{u},e\rightarrow e'}, e')
\end{align}
where
\begin{align}
&G_{e\rightarrow e'}=c^{i}_{2}\bm{\geq} \underline{t}_{S^{i}_{e}}\label{eqn:TGA_guard_uncontrol}\\
&R_{\tn{u},e\rightarrow e'}=c^{i}\label{eqn:TGA_reset_uncontrol}
\end{align}
\end{subequations}
and $\underline{t}_{S^{i}_{e}}\in\mathds{R}_{\geq0}$ is a lower bound on the time for staying in $S^{i}_{e}$ and $v[R_{\tn{u},e\rightarrow e'}]$ is defined in  \eqref{eqn:clock_valuation_update_extended} with $\alpha=(0,0)$ and all entries of $\beta$ equal to zero. Note that $\sigma_{\tn{u}}^{i}$ is the action on the transition $\delta_{e\rightarrow e'}$, as $e$ and $e'$ are generated using different slices from the $i^{\tn{th}}$ slice-family.

At each location $e$, where the control $g^{j}(x)$ is applied, the following transitions are defined for all $i\in\{1,\dots,$\\$|K_{U}|\}\backslash\{j\}$
\begin{subequations}
\begin{align}
&\delta_{e\rightarrow e'} = (e, G_{e\rightarrow e'}, \sigma_{\tn{c}}^{i}, R_{\tn{c},e\rightarrow e'}, e'),
\end{align}
where the control $g^{i}(x)$ is applied in $e'$ and
\begin{align}
&G_{e\rightarrow e'}=\{c\bm{\geq} 0| c\in \{c^{i}_{2}|i\in\bm{k}\}\}\label{eqn:TGA_guard_control}\\
&R_{\tn{c},e\rightarrow e'}=C.\label{eqn:TGA_reset_control}
\end{align}
\end{subequations}
Note that there are no active guard conditions and that the exact values of $\alpha,\beta$ are provided in Theorem~\ref{prop:suf_cond_soundness} in Section~\ref{sec:partitioning_ss} to obtain soundness of the abstraction.
\end{itemize}
\end{procedure}

For convenience the following notion is introduced.
\begin{definition}
Let $\mathcal{S}$ be a collection of slice-families, i.e., $\mathcal{S}=\{\mathcal{S}^{i}|i\in\bm{k}\}$. Then $\mathcal{G}\left(\mathcal{S}\right)$ is the timed game automaton generated by $K(\mathcal{S})$ according to Procedure~\ref{procedure:generation_of_TGA}.
\end{definition}

\begin{remark}
Nonetheless, the locations of $\mathcal{G}(\mathcal{S})$ are associated with cells of $K_{X}(\mathcal{S})$, we will also utilize the timed game automaton $\mathcal{G}_{\tn{ex}}(\mathcal{S})$ with locations associated to extended cells, i.e., (recall the definition of $\mathcal{Y}$ from Definition~\ref{def:extended_cell})
\begin{align}
E=\{e_{\tn{ex},y}|y\in\mathcal{Y}\}\times K_{U}.\label{eqn:locations_from_extended_cells}
\end{align}
The other steps of Procedure~\ref{procedure:generation_of_TGA} are identical for the two timed game automata $\mathcal{G}(\mathcal{S})$ and $\mathcal{G}_{\tn{ex}}(\mathcal{S})$.
\end{remark}


\section{Properties of the Abstraction}\label{sec:partitioning_ss}
The purpose of this section is to derive conditions for the partitions of the state space and control, and the conditions for the update maps under which an abstraction is sound.

To derive properties of the closed-loop system, we need a notion of controlled system, similar to $\mathcal{G}_{\kappa}$ from Definition~\ref{def:closed_loop_game}.
\begin{definition}
Let the control system $\Gamma=(X,U,f)$ be controlled using the strategy $\kappa$ be denoted $\Gamma_{\kappa}$. Then the dynamics of $\Gamma_{\kappa}$ is given by
\begin{subequations}
\begin{align}
\dot{x}&=f(x,g(x))\tn{ where}\\
g(x)&=g^{i}(x)\,\forall x\in e_{j}\tn{ iff }\kappa(e_{j})=\sigma^{i}_{\tn{c}}.
\end{align}
\end{subequations}
\end{definition}
\begin{remark}
We assume that Filippov solutions do not occur, as the implementation of the timed automaton is based on integers in the guards and invariants. For more details about the problems that can occur when having Filippov solutions, see \cite{On_formalism_and_stability_of_switched_systems}.
\end{remark}

A useful abstraction preserves safety. Therefore, the following is defined \cite{Progress_on_Reachability_Analysis_of_Hybrid_Systems_Using_Predicate_Abstraction}.
\begin{definition}[Sound Abstraction]\label{def:sound_abstraction}
 $ $ Let $ $ $\Gamma=(X,$\\$U,f)$ be a control system and suppose its state space $X$ is partitioned by $K_{X}(\mathcal{S})=\{e_{i}|i\in\bm{k}\}$ and its control is partitioned by $K_{U}$. Let the initial states be $X_{0}=\bigcup_{i\in\mathcal{I}}e_{i}$, with $\mathcal{I}\subseteq\bm{k}$.
Then a timed game automaton $\mathcal{G}=(E=K_{X}\times K_{U}, E_{0}, C, \Sigma_{\tn{c}}, \Sigma_{\tn{u}}, I, \Delta)$ with $E = K_{X}(\mathcal{S}) \times K_{U}$ and $E_{0}=\{e_{i}|i\in\mathcal{I}\}\times K_{U}$ is said to be a sound abstraction of $\Gamma$ on $[t_{1}, t_{2}]$ if $\forall t\in[t_{1}, t_{2}]$ and any strategy $\kappa$
\begin{subequations}
\begin{align}
&e_{i}\cap\tn{Reach}_{[t,t]}(\Gamma_{\kappa},X_{0})\neq\emptyset \tn{ implies}\\
&\notag\exists e_{0}\in E_{0}\tn{ such that}\\
&e_{i}\in\phi_{\mathcal{G}_{\kappa}}(t,e_{0}).
\end{align}
\end{subequations}
Note that the systems apply the same control at all time, as they follow the same strategy.
\end{definition}
If a sound abstraction $\mathcal{G}$ is safe for some strategy $\kappa$ then $\Gamma$ is also safe for the same strategy, as $\mathcal{G}_{\kappa}$ reaches all locations reached by $\Gamma_{\kappa}$.
\begin{definition}[Complete Abstraction]\label{def:complete_abstraction}
Let $\Gamma$ be a control system and suppose its state space $X$ is partitioned by $K_{X}(\mathcal{S})=\{e_{i}|i\in\bm{k}\}$ and its control is partitioned by $K_{U}$ and let the initial states be $X_{0}=\bigcup_{i\in\mathcal{I}}e_{i}$, with $\mathcal{I}\subseteq\bm{k}$.
Then a timed game automaton $\mathcal{G}=(E, E_{0}, C, \Sigma_{\tn{c}}, \Sigma_{\tn{u}}, I, \Delta)$ with $E = K_{X}(\mathcal{S}) \times K_{U}$ and $E_{0}=\{e_{i}|i\in\mathcal{I}\}\times K_{U}$ is said to be a complete abstraction of $\Gamma$ on $[t_{1}, t_{2}]$ if it is a sound abstraction and $\forall t\in[t_{1}, t_{2}]$, any $\kappa$, and
\begin{subequations}
\begin{align}
&\tn{for each }e_{i}\in\tn{Reach}_{[t,t]}(\mathcal{G}_{\kappa},E_{0})\\
&\notag\exists x_{0}\in X_{0}\tn{ such that}\\
&\phi_{\Gamma_{\kappa}}(t,x_{0})\in e_{i}.
\end{align}
\end{subequations}
\end{definition}
A complete abstraction $\mathcal{G}_{\kappa}$ is safe (unsafe) if and only if $\Gamma_{\kappa}$ is also safe (unsafe).

This next proposition follows directly from Proposition~6 in \cite{CDC2010}.
\begin{proposition}\label{prop:par_also_sound_complete}
$ $ A $ $ timed $ $ game $ $ automaton $ $ $\mathcal{G}_{\tn{ex}}(\mathcal{S})=$\\$\mathcal{G}_{1}(\mathcal{S}^{1})||\dots||\mathcal{G}_{k}(\mathcal{S}^{k})$, with locations abstracting extended cells, is a sound (complete) abstraction of the control system $\Gamma$ if and only if $\mathcal{G}_{1}(\mathcal{S}^{1}),\dots,\mathcal{G}_{k}(\mathcal{S}^{k})$ are sound (complete) abstractions of $\Gamma$.
\end{proposition}

Let $\cal S$ be a slice family. We say that a control $g$ is an admissible control if for each slice $S \in \cal S$ we have either $\dot{\varphi}_{g} (x) > 0$ for all $x \in S\backslash\tn{Cr}(f_{g})$ or $\dot{\varphi}_{g} (x) < 0$ for all $x \in S\backslash\tn{Cr}(f_{g})$. We introduce the notation  $\dot{\varphi}_{g} > 0$ on $S$ if and only if  $\dot{\varphi}_{g} (x) > 0$ for some, thus, for all $x \in S\backslash\tn{Cr}(f_{g})$.

\begin{lemma}\label{lem:sound}
Let $\cal S$ be a slice family on $\mathbb R^n$, and $\varphi$ be a partitioning function associated to $\cal S$.
Let  $\{t_{j}\in\mathds{R}_{\geq0}|~j\in \bm{k}\cup\{0\}\}$ be a sequence of nondecreasing real numbers, and $\{g_j: \mathbb R^m \to \mathbb R^n|~j\in \bm{k}\}$ be a sequence of controls, where $g_{j}$ is applied for $t\in[t_{j-1},t_{j}]$.

Suppose $S \in \cal S$. For convenience, let $\underline{\dot{\varphi}}_{g_{j}}\equiv\inf\{|\dot{\varphi}_{g_j}(x)|\in\mathds{R}_{\geq0}|x\in S\}$, $\overline{\dot{\varphi}}_{g_{j}}\equiv\sup\{|\dot{\varphi}_{g_j}(x)|\in\mathds{R}_{\geq0}|x\in S\}$, $\Delta_{k}(t)\equiv t-t_{k-1}$, $J^{+}\equiv\{j\in\bm{k}|\dot{\varphi}_{g_{j}}>0\}$, and $J^{-}\equiv \bm{k}\backslash J^{+}$.

If for all $j\in\bm{k}$
\begin{align}
\tn{Reach}_{[t_{j-1},t_{j}]}(\Gamma_{g_{j}},x_{j-1})\subset S
\end{align}
then
for all $t\in[t_{k-1},t_{k}]$
\begin{align}
&\notag\sum_{j\in J^{+}\backslash\{k\}}  \Delta_{j}(t_{j})\underline{\dot{\varphi}}_{g_{j}}-\sum_{j\in J^{-}\backslash\{k\}}\Delta_{j}(t_{j})\overline{\dot{\varphi}}_{g_{j}}\\
&\notag+\begin{cases}\Delta_{k}(t)\underline{\dot{\varphi}}_{g_{k}}\tn{ if }\dot{\varphi}_{g_{k}}>0\\
-\Delta_{k}(t)\overline{\dot{\varphi}}_{g_{k}}\tn{ if }\dot{\varphi}_{g_{k}}<0\end{cases}\\
&\notag\leq\sum_{j=1}^{k-1}\int_{t_{j-1}}^{t_{j}}\dot{\varphi}_{g_{j}}(\phi_{\Gamma_{g_{j}}}(\Delta_{j}(\tau),x_{j-1}))d\tau \\
&\notag+ \int_{t_{k-1}}^{t}\dot{\varphi}_{g_{k}}(\phi_{\Gamma_{g_{k}}}(\Delta_{k}(\tau),x_{k-1}))d\tau\\
&\notag\leq\sum_{j\in J^{+}\backslash\{k\}}\Delta_{j}(t_{j})\overline{\dot{\varphi}}_{g_{j}}-\sum_{j\in J^{-}\backslash\{k\}}\Delta_{j}(t_{j})\underline{\dot{\varphi}}_{g_{j}} \\
&+
\begin{cases}
\Delta_{k}(t)\overline{\dot{\varphi}}_{g_{k}} \tn{ if }\dot{\varphi}_{g_{k}}>0\\
-\Delta_{k}(t)\underline{\dot{\varphi}}_{g_{k}}\tn{ if }\dot{\varphi}_{g_{k}}<0
\end{cases}.\label{eqn:term3}
\end{align}
Note that $x_{j}\equiv\phi_{\Gamma_{g_{j}}}(\Delta_{j}(t_{j}),x_{j-1})$.
\end{lemma}
\begin{proof}
The inequalities in $\eqref{eqn:term3}$ are the consequence of $\underline{\dot{\varphi}}_{g_j} \leq \dot{\varphi}_{g_j}(x)$ and $\overline{\dot{\varphi}}_{g_j} \geq \dot{\varphi}_{g_j}(x)$ for all $x \in S$.
\end{proof}
The principle of the lemma is illustrated in Figure~\ref{fig:lemma_principle}, where $\varphi(\phi_{\Gamma}(\Delta_{j}(\tau),x_{j-1}))$ for $j=1,\dots,4$
(black line) is plotted together with the upper approximation (red) and lower approximation (blue). It is seen that the inaccuracy of the approximation increases with time.
\begin{figure}[!htb]
    \centering
       \includegraphics[width=\linewidth]{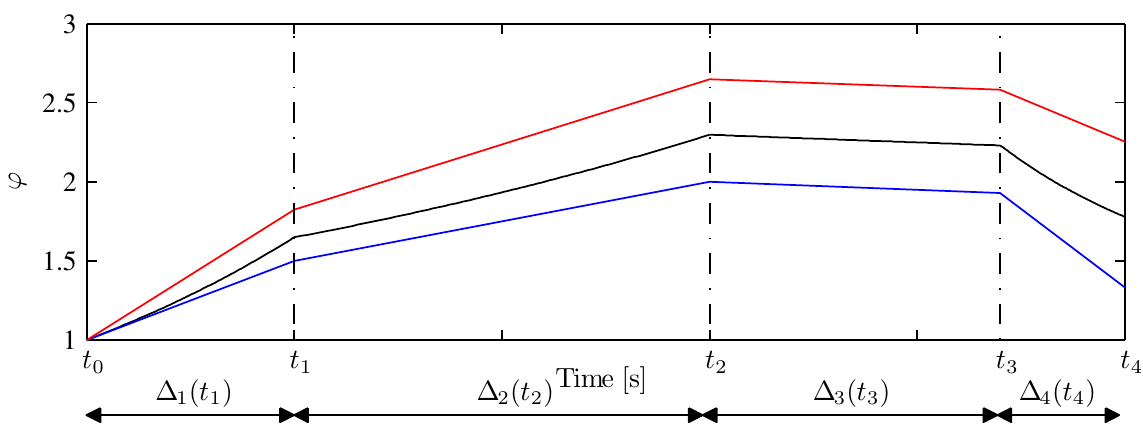}
    \caption{The value of the Lyapunov function evaluated along a solution curve (black) and an upper approximation (red) and a lower approximation (blue) of it.}
    \label{fig:lemma_principle}
\end{figure}

We use Lemma~\ref{lem:sound} to set up invariants and guards for the solution to stay in a slice $S$. Suppose $\varphi^{-1}([1, 3])$ in Figure~\ref{fig:lemma_principle} is a slice, then Corollary~\label{cor:over-approx} determines the time, when the red and blue lines intersect $\varphi^{-1}(1)$ and $\varphi^{-1}(3)$.

\begin{corollary}\label{cor:over-approx}
Let $\cal S$ be a slice family on $\mathbb R^n$, and $\varphi$ be a partitioning function associated to $\cal S$. Suppose $S\in\mathcal{S}$ and $S=\varphi^{-1}([a_{h-1},a_{h}])$, and define $\Delta a\equiv a_{h}-a_{h-1}$. Furthermore, assume that $x_{0}\in\varphi^{-1}(a_{h-1})$ and that $\dot{\varphi}_{g_{1}}>0$. Then from \eqref{eqn:term3} it follows that for all $t\in[t_{k-1},t_{k}]$ with $k\in J^{+}$, an invariant for the solution to stay in the slice $S$ is
\begin{align}
\sum_{j\in J^{+}\backslash\{k\}}\Delta_{j}(t_{j})\underline{\dot{\varphi}}_{g_{j}}-\sum_{j\in J^{-}\backslash\{k\}}\Delta_{j}(t_{j})\overline{\dot{\varphi}}_{g_{j}}+\Delta_{k}(t)\underline{\dot{\varphi}}_{g_{k}}\leq\Delta a\label{eqn:inv_s1}.
\end{align}
This implies that if inequality \eqref{eqn:inv_s1} is violated, then $\phi_{\Gamma_{g_{k}}}(t-t_{k-1},x_{k-1})\notin S$. Similarly, for all $t\in[t_{k-1},t_{k}]$, with $k\in J^{+}$ a guard for the solution to stay in the slice is
\begin{align}
\sum_{j\in J^{+}\backslash\{k\}}\Delta_{j}(t_{j})\overline{\dot{\varphi}}_{g_{j}}-\sum_{j\in J^{-}\backslash\{k\}}\Delta_{j}(t_{j})\underline{\dot{\varphi}}_{g_{j}}+\Delta_{k}(t)\overline{\dot{\varphi}}_{g_{k}}\geq\Delta a\label{eqn:gua_s1}.
\end{align}
This implies that if \eqref{eqn:gua_s1} is violated, then $\phi_{\Gamma_{g_{k}}}(t-t_{k-1},x_{k-1})$\\$\in S$.

Additionally, for all $t\in[t_{k-1},t_{k}]$, with $k\in J^{-}$ an invariant for the solution to stay in the slice is
\begin{align}
\sum_{j\in J^{+}\backslash\{k\}}\Delta_{j}(t_{j})\overline{\dot{\varphi}}_{g_{j}}-\sum_{j\in J^{-}\backslash\{k\}}\Delta_{j}(t_{j})\underline{\dot{\varphi}}_{g_{j}}-\Delta_{k}(t)\underline{\dot{\varphi}}_{g_{k}}\geq0\label{eqn:inv_s2}.
\end{align}
Finally, for all $t\in[t_{k-1},t_{k}]$, with $k\in J^{-}$ a guard for the solution to stay in the slice is
\begin{align}
\sum_{j\in J^{+}\backslash\{k\}}\Delta_{j}(t_{j})\underline{\dot{\varphi}}_{g_{j}}-\sum_{j\in J^{-}\backslash\{k\}}\Delta_{j}(t_{j})\overline{\dot{\varphi}}_{g_{j}}-\Delta_{k}(t)\overline{\dot{\varphi}}_{g_{k}}\leq0\label{eqn:gua_s2}.
\end{align}
\end{corollary}
\begin{proof}
Note that if $\dot{\varphi}_{g_{k}}>0$, the solution leaves the slice at some $t\in[t_{k-1},t_{k}]$ when $\phi_{\Gamma_{g_{j}}}(t-t_{k-1},x_{k-1})\in\varphi^{-1}(a_{h})$, i.e., for some $t\in[t_{k-1},t_{k}]$
\begin{align}
\notag&\sum_{j=1}^{k-1}\int_{t_{j-1}}^{t_{j}}\dot{\varphi}_{g_{j}}(\phi_{\Gamma_{g_{j}}}(\Delta_{j}(\tau),x_{j-1}))d\tau\\
&+ \int_{t_{k-1}}^{t}\dot{\varphi}_{g_{k}}(\phi_{\Gamma_{g_{k}}}(\Delta_{k}(\tau),x_{k-1}))d\tau
=\Delta a.
\end{align}
If $\dot{\varphi}_{g_{k}}<0$, the solution leaves the slice when $\phi_{\Gamma_{g_{k}}}(\Delta_{k}(t),$\\$x_{k-1})\in\varphi^{-1}(a_{h-1})$ for some $t\in[t_{k-1},t_{k}]$, i.e., for some $t\in[t_{k-1},t_{k}]$
\begin{align}
\notag&\sum_{j=1}^{k-1}\int_{t_{j-1}}^{t_{j}}\dot{\varphi}_{g_{j}}(\phi_{\Gamma_{g_{j}}}(\Delta_{j}(\tau),x_{j-1}))d\tau\\
&+ \int_{t_{k-1}}^{t}\dot{\varphi}_{g_{k}}(\phi_{\Gamma_{g_{k}}}(\Delta_{k}(\tau),x_{k-1}))d\tau=0.
\end{align}
This provides the right hand sides of \eqref{eqn:inv_s1}-\eqref{eqn:gua_s2}. However, note that $\leq$ $(\geq)$ also changes to $\geq$ $(\leq)$, which is due to the changing direction of $\dot{\varphi}_{g_{j}}$.
\end{proof}

The corollary provides guard and invariant conditions, which give the minimum and maximum time a trajectory stays within a slice for a given sequence of controls.

A sufficient and necessary condition for soundness of an abstraction is formulated in the following.
To stress that the control is of importance, we denote $S^{i}_{e}$, used in \eqref{eqn:TGA_invariants} and \eqref{eqn:TGA_guard_uncontrol}, by $S^{i}_{(y_{i},g)}$.
\begin{theorem}
\label{prop:suf_cond_soundness}
A timed game automaton $\mathcal{G}(\mathcal{S})=(E, E_{0}, C,$\\$ \Sigma_{\tn{c}},\Sigma_{\tn{u}}, I, \Delta)$ is a sound abstraction of the control system $\Gamma$, if and only if
its invariants and guards are given by \eqref{eqn:TGA_invariants} and \eqref{eqn:TGA_guard_uncontrol}, where for each $y\in \mathcal{Y}$ and $g\in K_{U}$
\begin{subequations}\label{eqn:suf_cond_soundness}
\begin{align}
\underline{t}_{S^{i}_{(y_{i},g)}}&\leq\frac{|a^{i}_{y_{i}}-a^{i}_{y_{i}-1}|}{\sup\{|\dot{\varphi}^{i}_{g}(x)|\in\mathds{R}_{\geq0}|x\in S^{i}_{y_{i}}\}}\\
\overline{t}_{S^{i}_{(y_{i},g)}}&\geq\frac{|a^{i}_{y_{i}}-a^{i}_{y_{i}-1}|}{\inf\{|\dot{\varphi}^{i}_{g}(x)|\in\mathds{R}_{\geq0}|x\in S^{i}_{y_{i}}\}}
\end{align}
\end{subequations}
and $\dot{\varphi}^{i}_{g}(x)$ is defined as shown in \eqref{eqn:Lyap_der}. 
The update map for transitions relations associated with controllable actions, see \eqref{eqn:TGA_reset_control}, between two locations $e$ and $e'$ with control $g$ respectively $g'$ is
\begin{align}
&\notag v[R_{\tn{c},e\rightarrow e'}]\equiv\\
&\begin{cases}
\begin{bmatrix}
\frac{\overline{t}_{S^{i}_{(y_{i},g')}}}{\overline{t}_{S^{i}_{(y_{i},g)}}}&0\\
0&\frac{\underline{t}_{S^{i}_{(y_{i},g')}}}{\underline{t}_{S^{i}_{(y_{i},g)}}}
\end{bmatrix}v(c^{i})
\tn{ if }\dot{\varphi}_{g}\dot{\varphi}_{g'}>0\\
\begin{bmatrix}
\overline{t}_{S^{i}_{(y_{i},g')}}\\
\underline{t}_{S^{i}_{(y_{i},g')}}
\end{bmatrix}-
\begin{bmatrix}
0&\frac{\overline{t}_{S^{i}_{(y_{i},g')}}}{\underline{t}_{S^{i}_{(y_{i},g)}}}\\
\frac{\underline{t}_{S^{i}_{(y_{i},g')}}}{\overline{t}_{S^{i}_{(y_{i},g)}}}&0
\end{bmatrix}v(c^{i})
\tn{ otherwise.}
\end{cases}\label{eqn:suf_cond_soundness_reset}
\end{align}
\end{theorem}
\begin{proof}
In this proof, we show by induction that the invariants and guards imposed on the clocks of $\mathcal{G}(\mathcal{S})$ generated using Theorem~\ref{prop:suf_cond_soundness}, where
\begin{subequations}\label{eqn:closest_sound}
\begin{align}
\underline{t}_{S^{i}_{(y_{i},g)}}&=\frac{|a^{i}_{y_{i}}-a^{i}_{y_{i}-1}|}{\sup\{|\dot{\varphi}^{i}_{g}(x)|\in\mathds{R}_{\geq0}|x\in S^{i}_{y_{i}}\}}\\
\overline{t}_{S^{i}_{(y_{i},g)}}&=\frac{|a^{i}_{y_{i}}-a^{i}_{y_{i}-1}|}{\inf\{|\dot{\varphi}^{i}_{g}(x)|\in\mathds{R}_{\geq0}|x\in S^{i}_{y_{i}}\}}
\end{align}
\end{subequations}
are equivalent to the guards and invariants given by Corollary~\ref{cor:over-approx}. As Corollary~\ref{cor:over-approx} gives conditions for a sound approximation this will prove the theorem.

Recall that $S^{i}_{y_{i}}=(\varphi^{i})^{-1}([a^{i}_{y_{i}-1},a^{i}_{y_{i}}])$, where $a^{i}_{y_{i}-1}<a^{i}_{y_{i}}$ and $\Delta a=a^{i}_{y_{i}}-a^{i}_{y_{i}-1}$. Assume that $x_{0}\in(\varphi^{i})^{-1}(a^{i}_{y_{i}-1})$, $\dot{\varphi}^{i}_{g_{1}}>0$, and $v_{0}(c^{i})=(0,0)$. Note that the valuation of the clocks are assumed to be zero, as $x_{0}$ is on the boundary of the considered slice.

First, we show that the invariants of Corollary~\ref{cor:over-approx} and Theorem~\ref{prop:suf_cond_soundness} are the same for all $t\in[t_{0},t_{1}]$; secondly, we assume they are the same for all $t\in[t_{k-1},t_{k}]$. Finally, we show that they are the same for all $t\in[t_{k},t_{k+1}]$.

Base case: We show that for $t\in[t_{0},t_{1}]$ the guards and invariants in \eqref{eqn:inv_s1}-\eqref{eqn:gua_s2}, shown in \eqref{eqn:base_case_cor_eq} for the considered case, are equivalent to the guard and invariant generated by using Theorem~\ref{prop:suf_cond_soundness}. From \eqref{eqn:inv_s1}-\eqref{eqn:gua_s2} we know that for all $t\in[t_{0},t_{1}]$
\begin{subequations}\label{eqn:base_case_cor_eq}
\begin{align}
&(t-t_{0})\underline{\dot{\varphi}}_{g_{1}}\leq\Delta a\Leftrightarrow(t-t_{0})\leq \overline{t}_{g_{1}}\\
&(t-t_{0})\overline{\dot{\varphi}}_{g_{1}}\geq\Delta a\Leftrightarrow(t-t_{0})\geq \underline{t}_{g_{1}}.
\end{align}
\end{subequations}
The invariants and guards of the abstraction are $c^{i}_{1}\leq\overline{t}_{g_{1}}$ and $c^{i}_{2}\geq\underline{t}_{g_{1}}$ where the clock valuations for all $t\in[t_{0},t_{1}]$ are
\begin{subequations}\label{eqn:base_case_prop_eq}
\begin{align}
&v(c^{i}_{1})\leq\overline{t}_{g_{1}}\Leftrightarrow v_{0}(c^{i}_{1})+(t-t_{0})\leq\overline{t}_{g_{1}}\\
&v(c^{i}_{2})\geq\underline{t}_{g_{1}}\Leftrightarrow v_{0}(c^{i}_{2})+(t-t_{0})\geq\underline{t}_{g_{1}}.
\end{align}
\end{subequations}
It is seen that \eqref{eqn:base_case_cor_eq} and \eqref{eqn:base_case_prop_eq} are equivalent, as $v_{0}(c^{i})=(0,0)$.

Inductive step: Assume that the guards and invariants of Theorem~\ref{prop:suf_cond_soundness} are equivalent to the guards and invariants derived from Corollary~\ref{cor:over-approx} for all $t\in[t_{k-1},t_{k}]$ and that $k\in J^{+}$
\begin{subequations}
\begin{align}
&\notag v(c^{i}_{1})\leq \overline{t}_{g_{k}}\Leftrightarrow\\
& \sum_{j\in J^{+}\backslash\{k\}}\Delta_{j}(t_{j})\underline{\dot{\varphi}}_{g_{j}}-\sum_{j\in J^{-}\backslash\{k\}}\Delta_{j}(t_{j})\overline{\dot{\varphi}}_{g_{j}}
+\Delta_{k}(t)\underline{\dot{\varphi}}_{g_{k}}\leq\Delta a\label{eqn:pos_inv}\\
&\notag v(c^{i}_{2})\geq \underline{t}_{g_{k}}\Leftrightarrow\\
& \sum_{j\in J^{+}\backslash\{k\}}\Delta_{j}(t_{j})\overline{\dot{\varphi}}_{g_{j}}-\sum_{j\in J^{-}\backslash\{k\}}\Delta_{j}(t_{j})\underline{\dot{\varphi}}_{g_{j}}
+\Delta_{k}(t)\overline{\dot{\varphi}}_{g_{k}}\geq\Delta a
\end{align}
\end{subequations}
and if $k\in J^{-}$, then
\begin{subequations}
\begin{align}
&\notag v(c^{i}_{1})\leq \overline{t}_{g_{k}}\Leftrightarrow\\
& \sum_{j\in J^{+}\backslash\{k\}}\Delta_{j}(t_{j})\overline{\dot{\varphi}}_{g_{j}}-\sum_{j\in J^{-}\backslash\{k\}}\Delta_{j}(t_{j})\underline{\dot{\varphi}}_{g_{j}}
-\Delta_{k}(t)\underline{\dot{\varphi}}_{g_{k}}\geq0\\
&\notag v(c^{i}_{2})\geq \underline{t}_{g_{k}}\Leftrightarrow\\
&\sum_{j\in J^{+}\backslash\{k\}}\Delta_{j}(t_{j})\underline{\dot{\varphi}}_{g_{j}}-\sum_{j\in J^{-}\backslash\{k\}}\Delta_{j}(t_{j})\overline{\dot{\varphi}}_{g_{j}}
-\Delta_{k}(t)\overline{\dot{\varphi}}_{g_{k}}\leq0\label{eqn:neg_gua}
\end{align}
\end{subequations}

Now we show that the guards and invariants imposed by Theorem~\ref{prop:suf_cond_soundness} are still equivalent to those of Corollary~\ref{cor:over-approx} at $j=k+1$. To shorten the proof, we only show this in two of the possible cases (\underline{$\dot{\varphi}_{g_{k}}\dot{\varphi}_{g_{k+1}}>0$ and $\dot{\varphi}_{g_{k}}>0$}, $\dot{\varphi}_{g_{k}}\dot{\varphi}_{g_{k+1}}>0$ and $\dot{\varphi}_{g_{k}}<0$, \underline{$\dot{\varphi}_{g_{k}}\dot{\varphi}_{g_{k+1}}<0$ and $\dot{\varphi}_{g_{k}}>0$}, $\dot{\varphi}_{g_{k}}\dot{\varphi}_{g_{k+1}}<0$ and $\dot{\varphi}_{g_{k}}<0$).

Assume that $\dot{\varphi}_{g_{k}}\dot{\varphi}_{g_{k+1}}>0$ and $\dot{\varphi}_{g_{k}}>0$. Then the invariant assumed in Theorem~\ref{prop:suf_cond_soundness} becomes for all $t\in[t_{k},t_{k+1}]$
\begin{subequations}
\begin{align}
\frac{\overline{t}_{g_{k+1}}}{\overline{t}_{g_{k}}}v_{k}(c^{i}_{1})+\Delta_{k}(t)\leq\overline{t}_{g_{k+1}}\\
\frac{\underline{\dot{\varphi}}_{g_{k}}}{\underline{\dot{\varphi}}_{g_{k+1}}}v_{k}(c^{i}_{1})+\Delta_{k}(t)\leq\overline{t}_{g_{k+1}}\label{eqn:inv_1_suf}
\end{align}
\end{subequations}
We use \eqref{eqn:pos_inv} to obtain $v_{k}(c^{i}_{1})$ by dividing the inequality by $\underline{\dot{\varphi}}_{g_{k}}$
\begin{align}
v_{k}(c^{i}_{1})=
\sum_{j\in J^{+}}\Delta_{j}(t_{j})\frac{\underline{\dot{\varphi}}_{g_{j}}}{\underline{\dot{\varphi}}_{g_{k}}}-\sum_{j\in J^{-}}\Delta_{j}(t_{j})\frac{\overline{\dot{\varphi}}_{g_{j}}}{\underline{\dot{\varphi}}_{g_{k}}}\label{eqn:clock_valuation_1_suf}
\end{align}
Inserting \eqref{eqn:clock_valuation_1_suf} into \eqref{eqn:inv_1_suf} and multiplying it by $\underline{\dot{\varphi}}_{g_{k+1}}$ yields for all $t\in[t_{k},t_{k+1}]$ (Note that $\Delta a=\overline{t}_{g_{k+1}}\underline{\dot{\varphi}}_{g_{k+1}}$)
\begin{align}
&\sum_{j\in J^{+}}\Delta t_{j}\underline{\dot{\varphi}}_{g_{j}}-\sum_{j\in J^{-}}\Delta t_{j}\overline{\dot{\varphi}}_{g_{j}}
+\underline{\dot{\varphi}}_{g_{k+1}}\Delta_{k}(t)\leq\Delta a.
\end{align}
This equals \eqref{eqn:pos_inv} for $j=k+1$.

Now assume that $\dot{\varphi}_{g_{k}}\dot{\varphi}_{g_{k+1}}<0$ and $\dot{\varphi}_{g_{k}}>0$. The guard generated by the assumption of Theorem~\ref{prop:suf_cond_soundness} becomes for all $t\in[t_{k},t_{k+1}]$
\begin{subequations}
\begin{align}
\underline{t}_{g_{k+1}}-\frac{\underline{t}_{g_{k+1}}}{\overline{t}_{g_{k}}}v_{k}(c^{i}_{1})+\Delta_{k}(t)\geq\underline{t}_{g_{k+1}}\\
\frac{\Delta a}{\overline{\dot{\varphi}}_{g_{k+1}}}-\frac{\underline{\dot{\varphi}}_{g_{k}}}{\overline{\dot{\varphi}}_{g_{k+1}}}v_{k}(c^{i}_{1})+\Delta_{k}(t)\geq\underline{t}_{g_{k+1}}.\label{eqn:gua_2_suf}
\end{align}
\end{subequations}
We use \eqref{eqn:pos_inv} to obtain $v_{k}(c^{i}_{1})$ by dividing the inequality by $\underline{\dot{\varphi}}_{g_{k}}$
\begin{align}
v_{k}(c^{i}_{1})=
\sum_{j\in J^{+}}\Delta_{j}(t_{j})\frac{\underline{\dot{\varphi}}_{g_{j}}}{\underline{\dot{\varphi}}_{g_{k}}}-\sum_{j\in J^{-}}\Delta_{j}(t_{j})\frac{\overline{\dot{\varphi}}_{g_{j}}}{\underline{\dot{\varphi}}_{g_{k}}}.\label{eqn:clock_valuation_2_suf}
\end{align}
Inserting \eqref{eqn:clock_valuation_2_suf} into \eqref{eqn:gua_2_suf} and multiplying it with $\overline{\dot{\varphi}}_{g_{k+1}}$ yields for all $t\in[t_{k},t_{k+1}]$ (Note that $\Delta a=\underline{t}_{g_{k+1}}\overline{\dot{\varphi}}_{g_{k+1}}$)
\begin{subequations}
\begin{align}
&\Delta a-\sum_{j\in J^{+}}\Delta_{j}(t_{j})\underline{\dot{\varphi}}_{g_{j}}\tn{+}\sum_{j\in J^{-}}\Delta_{j}(t_{j})\overline{\dot{\varphi}}_{g_{j}}
\tn{+}\overline{\dot{\varphi}}_{g_{k+1}}\Delta_{k}(t)\geq\Delta a\\
&\sum_{j\in J^{+}}\Delta_{j}(t_{j})\underline{\dot{\varphi}}_{g_{j}}-\sum_{j\in J^{-}}\Delta_{j}(t_{j})\overline{\dot{\varphi}}_{g_{j}}-\overline{\dot{\varphi}}_{g_{k+1}}\Delta_{k}(t)\leq0
\end{align}
\end{subequations}
This equals \eqref{eqn:neg_gua} for $j=k+1$.
\end{proof}

\begin{remark}
The $ $ theorem $ $ provides $ $ the $ $ closest $ $ possible\\sound abstraction using only functions with constant derivatives, when \eqref{eqn:closest_sound} is satisfied. This follows from Figure~\ref{fig:lemma_principle}, where it is seen that the lower and upper approximations can actually be a tangent to the Lyapunov function in the beginning of a time interval. Hence, if the conditions are strengthened, the abstraction will no longer be a sound approximation of the control system.
\end{remark}

\section{Illustrative Example}\label{sec:results}

In this section, we apply the abstraction on a model of a unicycle given as
\begin{align}
\dot{x}&=
\begin{bmatrix}
0&1&0&0\\
0&0&0&0\\
0&0&0&1\\
0&0&0&0\end{bmatrix}x+
\begin{bmatrix}
0&0\\
1&0\\
0&0\\
0&1\end{bmatrix}u
\end{align}
where $x\in\mathds{R}^{4}$ and $u\in\mathds{R}^{2}$. The first coordinate $x_{1}$ is the position of the unicycle in the $x$-direction, $x_{2}$ is the velocity of the unicycle in the $x$-direction, $x_{3}$ is the position of the unicycle in the $y$-direction, and $x_{4}$ is the velocity of the unicycle in the $y$-direction. The inputs are the acceleration in the $x$-direction ($u_{1}$) and the acceleration in the $y$-direction ($u_{2}$).
The dynamics of this system is not complex contrary to the control objective is. The objective is to design a controller, ensuring that the system always reaches the goal set (green) from the initial set (blue), without hitting any of the obstacles (red), see Figure~\ref{fig:ill_ex}.
\begin{figure}[!htb]
    \centering
       \includegraphics[scale=1]{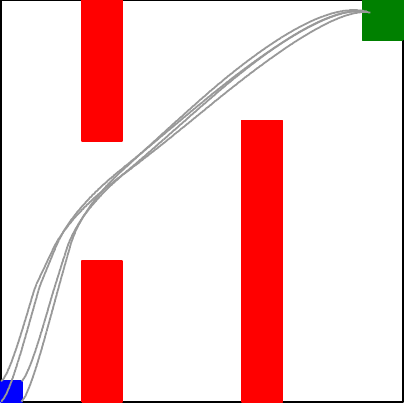}
    \caption{Track of the unicycle, given by initial set (blue), goal set (green), and obstacles (red).  The gray lines are trajectories of the robot, controlled using the proposed strategy.}
    \label{fig:ill_ex}
\end{figure}

The partition is conducted using four Lyapunov functions, which are not shown, as they are four-dimensional.

A control law, given by a switched controller, is generated from the abstraction, where the switching is determined by the cell that the system is in. From this example, it is concluded that it is possible to generate control strategies from the proposed abstraction. However, it is not automatically generated, as the update maps introduced in this paper are not implemented in currently available verification tools.

\section{Conclusion}\label{sec:conclusion}
In this paper, a method for abstracting control systems by timed games has been presented.
The method is based on partitioning the state space of the systems by set-differences of sets that are positive or negative invariant for all admissible controls. The timed games used in the abstraction have a more expressive update map than the update map usually allowed for timed game automata. This makes it possible to generate the proposed sound approximations, but has the consequence that no tools exists for automatic controller synthesis.

To enable synthesis of a control strategy that ensures safety of the system, conditions for soundness have been set up. These conditions tell when a sound abstraction can be realized from a partition of a state space and a partition of the control.
Finally, an example is provided to demonstrate that the formalism can be used to synthesize controllers that satisfy temporal specifications.

\bibliographystyle{abbrv}
\bibliography{bibliography}
\end{document}